\documentclass[letterpaper, 10 pt, conference]{ieeeconf}
\IEEEoverridecommandlockouts
\usepackage{graphicx,subcaption,epstopdf,amsmath,amssymb,mathtools}
\usepackage{mathptmx}
\usepackage{hyperref,url}
\hypersetup{colorlinks=true,
citecolor=blue,linkcolor=blue
}

\newtheorem{theorem}{Theorem}
\newtheorem{proposition}{Proposition}
\newcommand{\myref}[1]{(\ref{#1})}
\newcommand\norm[1]{\left\lVert#1\right\rVert}

\title{\LARGE \bf
Attitude Tracking Control for Aerobatic Helicopters: A Geometric Approach}

\author{Nidhish Raj\textsuperscript{1}, Ravi N. Banavar\textsuperscript{2}, Abhishek\textsuperscript{3}, and Mangal Kothari\textsuperscript{3}%
\thanks{\textsuperscript{1} Doctoral student, IIT Kanpur, currently Visiting Researcher, IIT Gandhinagar, India {\tt\small nraj@iitk.ac.in}} 
\thanks{ \textsuperscript{2} Professor, Systems and Control Engineering, IIT Bombay, currently Visiting Professor, Department of Electrical Engineering, IIT Gandhinagar, India {\tt\small ravi.banavar@gmail.com}}
\thanks{ \textsuperscript{3} Assistant Professor, Department of Aerospace Engineering, IIT Kanpur, India {\tt\small \{abhish,mangal\}@iitk.ac.in}}
}

\begin{document}

\maketitle
\thispagestyle{empty}
\pagestyle{empty}

\begin{abstract}

We consider the problem of attitude tracking for small-scale aerobatic helicopters. A small scale helicopter has two subsystems: the fuselage, modeled as a rigid body; and the rotor, modeled as a first order system. Due to the coupling between rotor and fuselage, the complete system does not inherit the structure of a simple mechanical system. The coupled rotor fuselage dynamics is first transformed to rigid body attitude tracking problem with a first order actuator dynamics. The proposed controller is developed using geometric and backstepping control technique. The controller is globally defined on $SO(3)$ and is shown to be locally exponentially stable. The controller is validated in simulation and experiment for a 10 kg class small scale flybarless helicopter by demonstrating aggressive roll attitude tracking.

\end{abstract}


\section{Introduction}
Small-scale conventional helicopters with a single main rotor and a tail rotor are capable of performing extreme 3D aerobatic maneuvers \cite{gavrilets2001aggressive}, \cite{abbeel2010autonomous}, \cite{gerig2008modeling}. Such maneuvers involve large angle rotation with high angular velocity, inverted flight, pirouette etc. This necessitates a tracking controller which is globally defined and is capable of achieving  fast rotational maneuvers.

 The attitude tracking problem of a helicopter is quite different from that of a rigid body. The control moments generated by the rotor excite the rigid body dynamics of the fuselage which in-turn affects the rotor loads and its dynamics causing nonlinear coupling. The key differences between the rigid body tracking problem and attitude tracking of a helicopter are the following: 1) the presence of large aerodynamic damping in the rotational dynamics; and 2) the required control moment for tracking cannot be applied instantaneously due to the rotor blade dynamics. The control moments are produced by the rotor subsystem which has a first order dynamics \cite{mettler2013identification}. The importance of including rotor dynamics in controller design for large scale helicopters has been extensively studied in the literature \cite{hall1973inclusion}, \cite{takahashi1994h}, \cite{ingle1994effects}, \cite{panza2014rotor}. Hall and Bryson \cite{hall1973inclusion} have shown the importance of rotor state feedback in achieving tight attitude control for large scale helicopters, while Takahashi \cite{takahashi1994h} compares $H_\infty$ attitude controller design for cases with and without rotor state feedback.
 
In this article, we propose an attitude tracking controller for small scale helicopters using notions based on geometric control and backstepping control design approaches. We show that the controller is defined globally on the attitude manifold, $SO(3)$, achieves local exponential stability and is capable of performing rapid rotational maneuvers. Previous attempts to small-scale helicopter attitude control are based on attitude parametrization such as Euler angles, which suffer from  singularity issues, or quaternions which have ambiguity in representation. The proposed controller being defined on $SO(3)$ is free of these issues. In \cite{ahmed2009flight}, an adaptive backstepping stabilizing controller using Euler angles for a small scale helicopter with servo and rotor dynamics is considered. Tang et al. \cite{tang2015attitude} explicitly consider the rotor dynamics and design stabilizing controller based on sliding mode technique using Euler angles and hence confined to small angle maneuvers.

The paper is organized as follows: Section 2 describes the rotor-fuselage dynamics of a small-scale helicopter. Section 3 presents an attitude tracking controller for a rigid body and later presents the proposed controller for helicopter rotor-fuselage dynamics. The efficacy of the proposed design is demonstrated through numerical simulation in Section 4 and it's experimentally validation is given in Section 5.


\section{Helicopter Model}

Unlike quadrotors, a helicopter modeled as a rigid body does not capture all the dynamics required for high bandwidth attitude control purposes. A coupled rotor-fuselage model of a small-scale helicopter is considered \cite{mettler2013identification}. The fuselage is modeled as a rigid body and the rotor as a first order system which generates the required control moments. The inclusion of the rotor model is crucial as it introduces aerodynamic damping in the system.

\begin{figure}[!htbp]
\begin{center}
\includegraphics[width=1.0\linewidth]{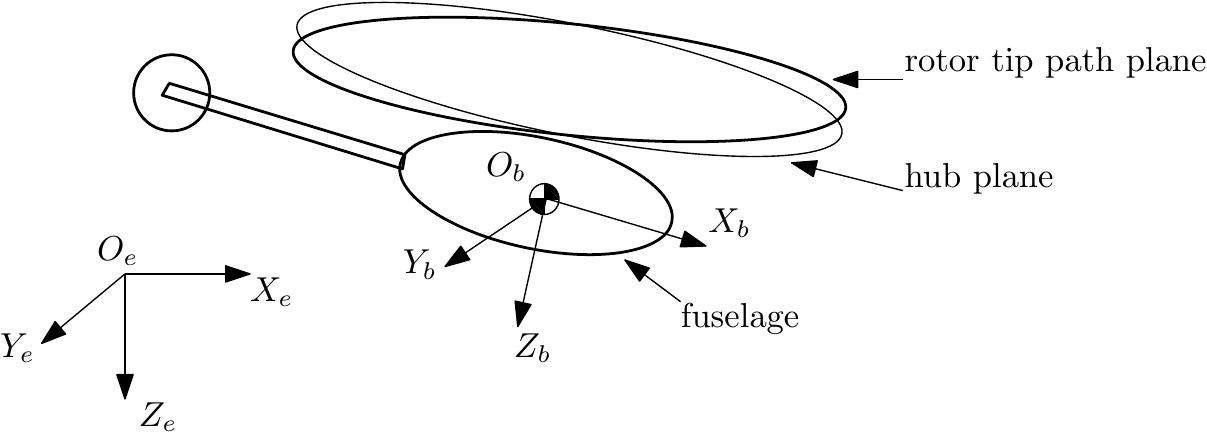}
\caption{Fuselage and tip path plane.}
\label{fig:heli_model}
\end{center}
\end{figure}

The rotational equations of motion of the fuselage are given by,
\begin{equation} \label{eq:rigid_body}
\begin{split}
\dot{R} = R\hat{\omega}, \\
J\dot{\omega} + \omega\times J\omega = M,
\end{split}
\end{equation}
where $R \in SO(3)$ is the rotation matrix which transforms vectors from the body fixed frame of reference, $(O_b,X_b,Y_b,Z_b)$, to a spatial frame of reference, $(O_e,X_e,Y_e,Z_e)$, $M = [M_x,M_y,M_z]$ is the external moment acting on the fuselage and $J$ is the body moment of inertia of the fuselage, $\omega = [\omega_x,\omega_y,\omega_z]$ is the angular velocity of the body frame with respect to the spatial frame expressed in the body frame. The hat operator, $\hat{\cdot}$, is a Lie algebra isomorphism from $\mathbb{R}^3$ to $\mathfrak{so(3)}$ given by
\begin{equation*}
\hat{\omega} = \begin{bmatrix}
0 & -\omega_z & \omega_y \\
\omega_z & 0 & -\omega_x \\
-\omega_y & \omega_x & 0 
\end{bmatrix}.
\end{equation*}

We consider here first order tip path plane (TPP) equations for the rotor as it captures the required dynamics for gross movement of fuselage \cite{mettler2013identification}. The rotor dynamics equations are given by

\begin{equation} \label{eq:flap_eq}
\begin{split}
\dot{a} = -\omega_y - a/\tau_m + \theta_a/\tau_m, \\
\dot{b} = -\omega_x - b/\tau_m + \theta_b/\tau_m, 
\end{split}
\end{equation}
where $a$ and $b$ are respectively the longitudinal and lateral tilt of the rotor disc with respect to the hub plane as shown in Fig. \ref{fig:rotor_fuselage}, $\tau_m$ is the main rotor time constant and $\theta_a$ and $\theta_b$ are the control inputs to the rotor subsystem. They are respectively the lateral and longitudinal cyclic blade pitch angles actuated by servos through a swashplate mechanism.

\begin{figure}[!htbp]
\begin{center}
\includegraphics[width=0.8\linewidth]{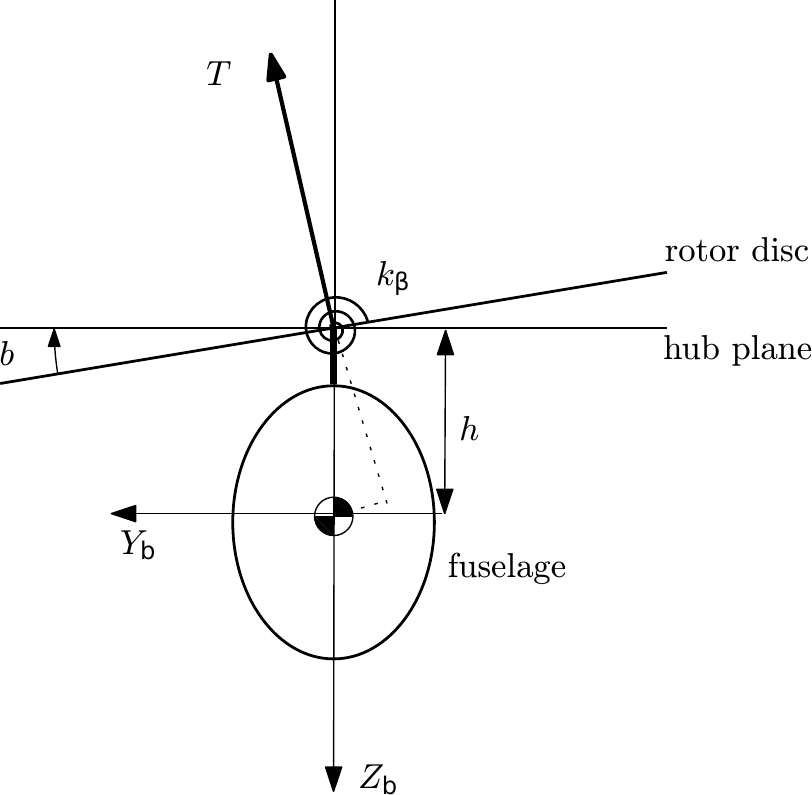}
\caption{Rotor-fuselage coupling.}
\label{fig:rotor_fuselage}
\end{center}
\end{figure}

The coupling of the rotor and fuselage occurs through the rotor hub. The rolling moment, $M_x$ and pitching moment $M_y$, acting on the fuselage due to the rotor flapping consists of two components -- due to tilting of the thrust vector, $T$, and due to the rotor hub stiffness, $k_\beta$,
\begin{equation*}
\begin{split}
M_x &= (hT + k_\beta)b, \\
 M_y &= (hT + k_\beta)a. 
\end{split}
\end{equation*}
Here $h$ is the distance of rotor hub from the center of mass. For near-hover condition the thrust can be considered constant which gives the equivalent hub stiffness, $K_\beta = (hT + k_\beta)$. The control moment about yaw axis, $M_z$, is applied through tail rotor which has a much faster aerodynamic response than the main rotor flap dynamics. The tail rotor along with the actuating servo is approximated as a first order system with $\tau_t$ as tail rotor time constant
\begin{equation*}
\dot{M}_z = -M_z/\tau_t + K_t\theta_t/\tau_t.
\end{equation*}

Since angular velocity of the fuselage is available for feedback, the main rotor dynamics \myref{eq:flap_eq} and tail rotor dynamics can be written in terms of control moments and a new control input $u$ as
\begin{equation} \label{eq:actuator}
\dot{M} = -AM + u,
\end{equation}
where A is a positive definite matrix defined as
\begin{equation*}
A \triangleq \begin{bmatrix}
1/\tau_m & 0 & 0 \\
0 & 1/\tau_m & 0 \\
0 & 0 & 1/\tau_t 
\end{bmatrix}
\end{equation*}
and the new control input $u$ is defined as
\begin{equation*}
u \triangleq \begin{bmatrix}
K_\beta(\theta_b/\tau_m - \omega_x) \\
K_\beta(\theta_a/\tau_m - \omega_y) \\
K_t\theta_t/\tau_t
\end{bmatrix}.
\end{equation*}

The combined rotor-fuselage dynamics given by \myref{eq:rigid_body} and \myref{eq:actuator} can be seen as a simple mechanical system driven by a force which has first order dynamics. The overall dynamical system does not have the form of a simple mechanical system \cite{bullo2004geometric} as the actuator dynamics is first order.


\section{Attitude Tracking Controller}

Given a twice differentiable attitude reference command $(R_d(t),\omega_d(t),\dot{\omega}_d(t))$, the objective is to design an attitude tracking controller for the helicopter. The combined rotor-fuselage dynamics is reproduced here for convenience

\begin{equation} \label{eq:fuselage}
\text{Fuselage}
\begin{cases}
& \dot{R} = R\hat{\omega}, \\
& J\dot{\omega} + \omega\times J\omega = M, 
\end{cases}
\end{equation}
\begin{equation} \label{eq:rotor}
\text{Rotor}
\begin{cases}
& \dot{M} = -AM + u.
\end{cases}
\end{equation}

First we design an attitude tracking controller for the fuselage, modeled as rigid body \myref{eq:fuselage}, using geometric control theory as is described in \cite{bullo2004geometric} and \cite{koditschek1989application}. Next, we use the results from this part to prove local exponential stability of the proposed helicopter tracking controller. The rigid body tracking controller has proportional derivative plus feed-forward components. The proportional action is derived from a tracking error function $\psi : SO(3) \times SO(3) \to \mathbb{R}$ which is defined in terms of the configuration error function $\psi_c : SO(3) \to \mathbb{R}$ as
\begin{equation*}
\psi(R,R_d) = \psi_c(R_d^TR) \coloneqq \frac{1}{2} tr[I - R_d^TR].
\end{equation*}
This is possible on a Lie group since the tracking problem can be reduced to a configuration stabilization problem about the identity because of the possibility of defining error between any two configurations using the group operation \cite{maithripala2006almost}. $\psi_c$ has a single critical point within the sub level set about the identity I, $\psi_c^{-1}(\leq 2,I) = \{R \in SO(3) |  \psi_c(R)<2 \}$. This sub level set represents the set of all rotations which are less than $\pi$ radians from the identity $I$. From the above function the attitude error vector, $e_R$, is defined as the differential of $\psi$ with respect to first argument,
\begin{equation*}
d_1\psi(R,R_d)\cdot R\hat{\omega}  = \frac{1}{2}[R_d^TR-R^TR_d]^\vee\cdot\omega,
\end{equation*}
\begin{equation*}
e_R = \frac{1}{2}[R_d^TR-R^TR_d]^\vee,
\end{equation*}
where $(\cdot)^\vee : \mathfrak{so(3)}\rightarrow\mathbb{R}^3$ is the inverse of hat map $\hat{(\cdot)}$. Since the velocities at reference and current configurations are in different tangent spaces they cannot be directly compared. Therefore $\dot{R}_d$ is transported to the tangent space at $R$ by the tangent map of the left action of $R^TR_d$. Thus, the tracking error for angular velocity is given by
\begin{equation*}
e_\omega = \omega - R^TR_d\omega_d.
\end{equation*}
The total derivative of $e_R$ is
\begin{equation*}
\begin{split}
\dot{e}_R &= \frac{1}{2}[-\hat{\omega}_d R_d^TR + R_d^TR\hat{\omega} + \hat{\omega}R^TR_d - R^TR_d\hat{\omega}_d]^\vee \\
&= \frac{1}{2}[R_d^TR(\hat{\omega} - R^TR_d\hat{\omega}_dR_d^TR) + (\hat{\omega} - R^TR_d\hat{\omega}_dR_d^TR)R^TR_d]^\vee \\
&= \frac{1}{2}[R_d^TR\hat{e}_\omega + \hat{e}_\omega R^TR_d]^\vee \\
&= B(R_d^TR)e_\omega,
\end{split}
\end{equation*}
where $B(R_d^TR) = \frac{1}{2}[tr(R^TR_d)I - R^TR_d]$ and $B(R_d^TR) \, <\, 1$ for all $R_d^TR \, \in \, SO(3)$. Here we have used the fact that $ [R\hat{x}R^T]^\vee = Rx$ for all $R \in SO(3)$ and $x \in \mathbb{R}^3$.
The time derivative of $e_\omega$ is
\begin{equation*}
\dot{e}_\omega = \dot{\omega} - R^TR_d\dot{\omega}_d + \hat{\omega}R^TR_d\omega_d.
\end{equation*}
The total derivative of $\psi$ is 
\begin{equation*}
\begin{split}
\frac{d\psi}{dt} &= -\frac{1}{2} tr(-\hat{\omega}_d R_d^T R + R_d^TR\hat{\omega}) \\
&= -\frac{1}{2} tr(R_d^TR(\hat{\omega} - R^TR_d\hat{\omega}_dR_d^TR)) \\
&= -\frac{1}{2} tr\left( \frac{1}{2}\left(R_d^TR - R^TR_d\right) \hat{e}_\omega\right) \\
&= e_R\cdot e_\omega.
\end{split}
\end{equation*}

$\psi_c$ is positive definite and quadratic within the sub level set $\psi_c^{-1}(\leq 2,I)$ which makes $\psi$ uniformly quadratic about the identity \cite{bullo2004geometric}. This implies there exist scalers $b_1$, $b_2$ such that $0 < b_1 \leq b_2$ and
\begin{equation*}
b_1\norm{e_R}^2 \leq \psi(R,R_d) \leq b_2\norm{e_R}^2 \quad \forall R,R_d \in \psi_c^{-1}(\leq 2,I).
\end{equation*}


\subsection{Attitude Tracking for Rigid body}
The tracking controller for the fuselage \cite{bullo2004geometric} based on the above error function is given by
\begin{equation} \label{eq:moment_des}
M_d = -k_R e_R - k_{\omega} e_{\omega} + \omega \times J\omega - J(\hat{\omega}R^TR_d\omega_d-R^TR_d \dot{\omega}_d)
\end{equation}
where $k_R$ and $k_\omega$ are positive constants, the third term cancels
the rotational dynamics, and the subsequent terms are the feedforward terms. The error dynamics for the rigid body can now be obtained by substituting the above desired moment, $M_d$, in \myref{eq:fuselage}, which results in
\begin{equation} \label{eq:att_err_dyn}
J\dot{e}_\omega = -k_Re_R - k_\omega e_\omega.
\end{equation}

The following theorem, taken from \cite{lee2010geometric}, shows exponential stability of the attitude tracking controller.

\begin{theorem}(Exponential stability of attitude error dynamics)
The control moment given in \myref{eq:moment_des} makes the equilibrium $(e_R,e_\omega) = (0,0)$ of tracking error dynamics defined in \myref{eq:att_err_dyn} exponentially stable for all initial conditions satisfying 
\begin{equation*}
k_R\ \psi(R(0),R_d(0)) + \frac{1}{2}\lambda_{max}(J)\norm{e_\omega (0)}^2 < 2k_R.
\end{equation*}
\end{theorem}
\begin{proof}
Define a Lyapunov candidate function for the error dynamics \myref{eq:att_err_dyn} 
\begin{equation*}
V_1 = \frac{1}{2}e_\omega \cdot J e_\omega + k_R\psi(R,R_d) + \epsilon e_R\cdot e_\omega,
\end{equation*}
where $0 < \epsilon \in \mathbb{R}$. $V_1$ can be lower and upper bounded by 
\begin{equation*}
\begin{split}
\frac{1}{2}\lambda_{min}(J)\norm{e_\omega}^2 + k_Rb_1\norm{e_R}^2 - \epsilon \norm{e_R}\norm{e_\omega} \leq V_1 \\
\leq \frac{1}{2}\lambda_{max}(J)\norm{e_\omega}^2 + k_Rb_2\norm{e_R}^2 + \epsilon \norm{e_R}\norm{e_\omega}
\end{split}
\end{equation*}
resulting in the relation,
\begin{equation} \label{eq:V}
z_1^TM_1z_1 \leq V_1 \leq z_1^TM_2z_1,
\end{equation}
where $z_1 = [\norm{e_\omega} \norm{e_R}]$ and
\begin{equation*}
M_1 = \begin{bmatrix}
\frac{\lambda_{min}(J)}{2} & -\frac{\epsilon}{2} \\
-\frac{\epsilon}{2} & k_Rb_1 
\end{bmatrix}, M_2 = \begin{bmatrix}
\frac{\lambda_{max}(J)}{2} & \frac{\epsilon}{2} \\
\frac{\epsilon}{2} & k_Rb_2 
\end{bmatrix}.
\end{equation*} 
The time derivative of $V_1$ is given by
\begin{equation*}
\begin{split}
\dot{V}_1 &= e_\omega\cdot J\dot{e}_\omega + k_Re_R\cdot e_\omega + \epsilon \dot{e_R}\cdot e_\omega + \epsilon e_R\cdot\dot{e}_\omega \\
&= -k_\omega \norm{e_\omega}^2 + \epsilon B(R_d^TR)e_\omega \cdot e_\omega \\
& \quad - \epsilon k_R e_R\cdot J^{-1}e_R - \epsilon k_\omega e_R \cdot J^{-1} e_\omega.
\end{split}
\end{equation*}
$\dot{V}_1$ can be upper bounded by 
\begin{equation*}
\begin{split}
\dot{V}_1 \leq -k_\omega \norm{e_\omega}^2 + \epsilon \norm{e_\omega}^2 - \frac{\epsilon k_R}{\lambda_{max}(J)} \norm{e_R}^2 \\
+ \frac{\epsilon k_\omega}{\lambda_{min}(J)}\norm{e_\omega}\norm{e_R},
\end{split}
\end{equation*}
which can be written as,
\begin{equation} \label{eq:V_dot}
\dot{V}_1 \leq -z_1^TW_1z_1,
\end{equation}
where,
\begin{equation*}
W_1 = \begin{bmatrix}
k_\omega - \epsilon & -\frac{\epsilon k_\omega}{2\lambda_{min}(J)} \\
-\frac{\epsilon k_\omega}{2\lambda_{min}(J)} & \frac{\epsilon k_R}{\lambda_{max}(J)} 
\end{bmatrix}.
\end{equation*}
Choosing $\epsilon$ such that
\begin{equation*}
\epsilon < \min \left\{k_\omega, \sqrt{2k_Rb_1\lambda_{min}(J)}, \frac{4k_Rk_\omega \lambda_{min}(J)^2}{k_\omega^2\lambda_{max}(J) + 4k_R \lambda_{min}(J)^2 } \right\}
\end{equation*}
makes the matrices $M_1$, $M_2$ and $W_1$ positive definite. This makes $V_1$ quadratic from  \myref{eq:V} and $\dot{V}_1$ negative definite as long as the configuration error $R_e(t)=R_d^T(t)R(t)$ remains in the sub level set $\psi_c^{-1}(\leq 2,I)$. This is shown to be true in the following sequence of arguments. Consider a Lyapunov candidate $V_2 = \frac{1}{2}e_\omega \cdot J e_\omega + k_R\psi(R,R_d)$ for the attitude error dynamics. $\dot{V}_2 = -k_\omega \norm{e_\omega}^2 \leq 0$.  This guarantees
\begin{equation*}
\begin{split}
k_R\psi(R(t),R_d(t)) \leq k_R\psi(R(t),R_d(t)) + \frac{1}{2}e_\omega(t) \cdot J e_\omega(t) \\
\leq k_R\psi(R(0),R_d(0)) + \frac{1}{2}e_\omega(0) \cdot J e_\omega(0) < 2k_R \\
\implies \psi(R(t),R_d(t)) < 2 \quad \forall t>0.
\end{split}
\end{equation*} 
Therefore there exists positive constants $\alpha_1$, $\beta_1$ such that $\psi(t) \leq \min\{2,\alpha_1e^{-\beta_1t}\}$.

\end{proof}


\subsection{Attitude Tracking for Helicopter}

In this subsection we bring in the rotor dynamics which induces the desired torque computed in the previous part through a first order subsystem. The error between this  desired torque and the applied torque on the fuselage is denoted as $e_M \triangleq M - M_d$. We derive the error dynamics for the combined rotor-fuselage dynamics. Equation \myref{eq:fuselage} can be rewritten as
\begin{equation*} 
J\dot{\omega} + \omega\times J\omega = M_d + e_M.
\end{equation*} 
Using $M_d$ for rigid body tracking from \myref{eq:moment_des} gives the error dynamics for fuselage as
\begin{equation} \label{eq:er_fuselage}
J\dot{e}_\omega = -k_Re_R - k_\omega e_\omega + e_M.
\end{equation}
Taking the derivative of $e_M$ and using \myref{eq:rotor} leads to the following error dynamics for rotor 
\begin{equation} \label{eq:er_rotor}
\dot{e}_M = -Ae_M -AM_d - \dot{M}_d + u.
\end{equation}
Equations \myref{eq:er_fuselage} and \myref{eq:er_rotor} constitute the error dynamics for the rotor-fuselage system. It is clear that the resulting error dynamics has a strict-feedback form wherein $e_M$ acts as a virtual control input in \myref{eq:er_fuselage}. Therefore a backstepping approach can be used for controller synthesis \cite{khalil2014nonlinear}.
We claim that the rotor-fuselage error dynamics is locally exponentially stable if the control input $u$ is chosen to be 
\begin{equation} \label{eq:control_ip}
u = \dot{M}_d + AM_d - e_\omega -\epsilon J^{-1}e_R.
\end{equation}
In the above expression, the derivative of the desired control moment is obtained by differentiating \myref{eq:moment_des},
\begin{equation} \label{eq:moment_dot}
\begin{split}
\dot{M}_d = -k_R\dot{e}_R - k_\omega\dot{e}_\omega + \dot{\omega}\times J\omega + \omega \times J\dot{\omega}  - J(\dot{\hat{\omega}}R^TR_d\omega_d \\
 - \hat{\omega}^2R^TR_d\omega_d + 2\hat{\omega}R^TR_d\dot{\omega}_d - R^TR_d\hat{\omega}_d\dot{\omega}_d - R^TR_d\ddot{\omega}_d).
\end{split}
\end{equation}

\begin{proposition}(Exponential stability of rotor-fuselage error dynamics)
The control input given in \myref{eq:control_ip} renders the equilibrium $(e_R,e_\omega,e_M) = (0,0,0)$ of the rotor fuselage error dynamics exponentially stable for all initial conditions satisfying
\begin{equation*}
\begin{split}
k_R\psi(R(0),R_d(0)) + \frac{1}{2}\lambda_{max}(J)\norm{e_\omega (0)}^2 + \frac{1}{2}\norm{e_M(0)}^2 \\
+ \epsilon\norm{e_R(0)}\norm{e_\omega(0)} < 2k_R.
\end{split}
\end{equation*}
\end{proposition}
\begin{proof}
Consider the following Lyapunov candidate for combined rotor-fuselage error dynamics
\begin{equation*}
\begin{split}
V &= V_1 + \frac{1}{2}\norm{e_M}^2 \\
&= \frac{1}{2}e_\omega \cdot J e_\omega + k_R\psi(R,R_d) + \epsilon e_R\cdot e_\omega + \frac{1}{2}\norm{e_M}^2.
\end{split}
\end{equation*}
$V$ is quadratic within the sub level set $\psi_c^{-1}(\leq 2,I)$ since $V_1$ is quadratic in the same set. 
The time derivative of $V$ is given by
\begin{equation*}
\begin{split}
\dot{V} &= e_\omega\cdot(-k_\omega e_\omega + e_M) + \epsilon \dot{e}_R\cdot e_\omega + \epsilon e_R\cdot J^{-1}(-k_Re_R \\ 
&-k_\omega e_\omega + e_M) + e_M\cdot \dot{e}_M\\
&= -k_\omega\norm{e_\omega}^2 + \epsilon  B(R_d^TR)e_\omega\cdot e_\omega - \epsilon k_R e_R\cdot J^{-1}e_R \\ &- \epsilon k_\omega e_R\cdot J^{-1}e_\omega \\ 
&+ e_M\cdot (\epsilon J^{-1}e_R + e_\omega -Ae_M -\dot{M}_d -AM_d + u) \\
\end{split}
\end{equation*}
From the previous subsection on rigid body tracking, the first four terms in the above expression have been rendered negative definite \myref{eq:V_dot}. By substituting $u$ from \myref{eq:control_ip} we get,
\begin{equation*}
\dot{V} = \dot{V}_1 - e_M\cdot Ae_M \leq -z_1^TW_1z_1 - e_M\cdot Ae_M  \leq -z^TWz,
\end{equation*}
where $z = [\norm{e_\omega} \norm{e_R} \norm{e_M}]^T$ and
\begin{equation*}
W = \begin{bmatrix}
k_\omega - \epsilon & -\frac{\epsilon k_\omega}{2\lambda_{min}(J)} & 0 \\
-\frac{\epsilon k_\omega}{2\lambda_{min}(J)} & \frac{\epsilon k_R}{\lambda_{max}(J)} & 0 \\
0 & 0 & \lambda_{min}(A) 
\end{bmatrix}
\end{equation*}
$V(t)$ remains quadratic when $R_e(t)=R_d(t)^TR(t)$ lies in the sub level set $\psi_c^{-1}(\leq 2,I)$. This is true since 
\begin{equation*}
\begin{split}
k_R\psi(R(t),R_d(t)) \leq k_R\psi(R(t),R_d(t)) + \frac{1}{2}e_\omega(t) \cdot J e_\omega(t) \\
+ \frac{1}{2}\norm{e_M(t)}^2 + \epsilon\norm{e_R(t)}\norm{e_\omega(t)} \\
\leq k_R\psi(R(0),R_d(0)) + \frac{1}{2}e_\omega(0) \cdot J e_\omega(0) \\
+ \frac{1}{2}\norm{e_M(0)}^2  + \epsilon\norm{e_R(0)}\norm{e_\omega(0)} < 2k_R \\
\implies \psi(R(t),R_d(t)) < 2 \quad \forall t>0.
\end{split}
\end{equation*}
$V(t)$ is positive definite, quadratic and $\dot{V}(t)$ is negative definite, therefore there exists positive scalars $\alpha$ and $\beta$ such that $\psi(t) \leq \min\{2, \alpha e^{-\beta t}\}$.

\end{proof}

Equation \myref{eq:moment_dot} implies that a feasible attitude reference trajectory must have a continuous second derivative of the angular velocity $\ddot{\omega}_d$ for a continuous control input. It is assumed that the fuselage body frame angular acceleration $\dot{\omega}$ is available for feedback. The flap angles $(a,b)$, which are difficult to measure, are not required for implementing the controller.


\section{Simulation Results}

The tracking controller given by \myref{eq:control_ip} was simulated for a 10 kg class model helicopter whose parameters are given in Table \ref{tab:heli_params}. The helicopter was given an initial attitude of 150 deg in roll angle and 57 deg/s of roll-rate and was subjected to a sinusoidal roll angle input with an amplitude of twenty degree and a frequency of one Hertz. Fig. \ref{fig:roll_response} shows the response and is evident that the controller is able to converge to reference command within one second. The controller is able to track the desired roll attitude with maximum flap deflection of $\pm$ $0.87$ degrees as shown in Fig. \ref{fig:flap_angles}. As expected, the longitudinal tilt of the rotor remains unchanged at zero as the maneuver simulated has purely lateral motion.

\begin{figure} [!htbp]
\centering
  \includegraphics[width=0.9\linewidth]{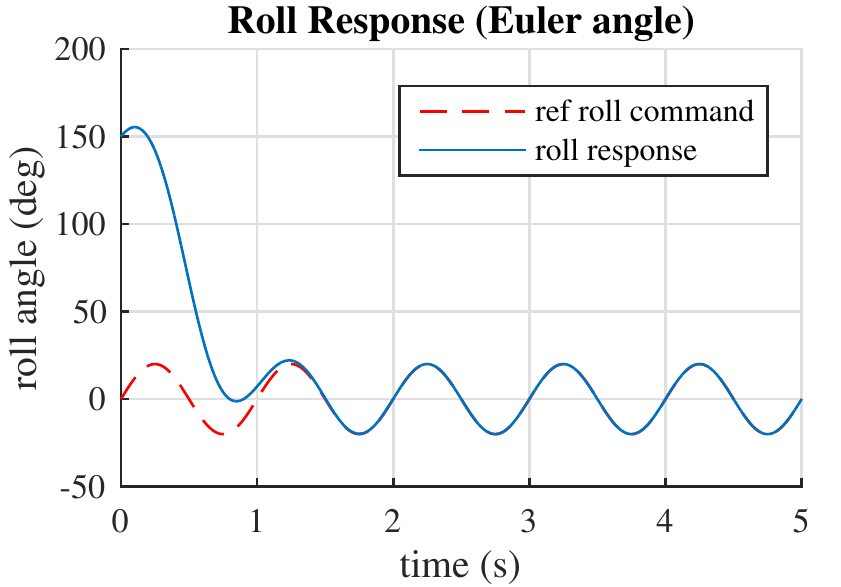}
  \caption{Roll tracking response.}
\end{figure}
\begin{figure} [!htbp]
  \centering
  \includegraphics[width=0.9\linewidth]{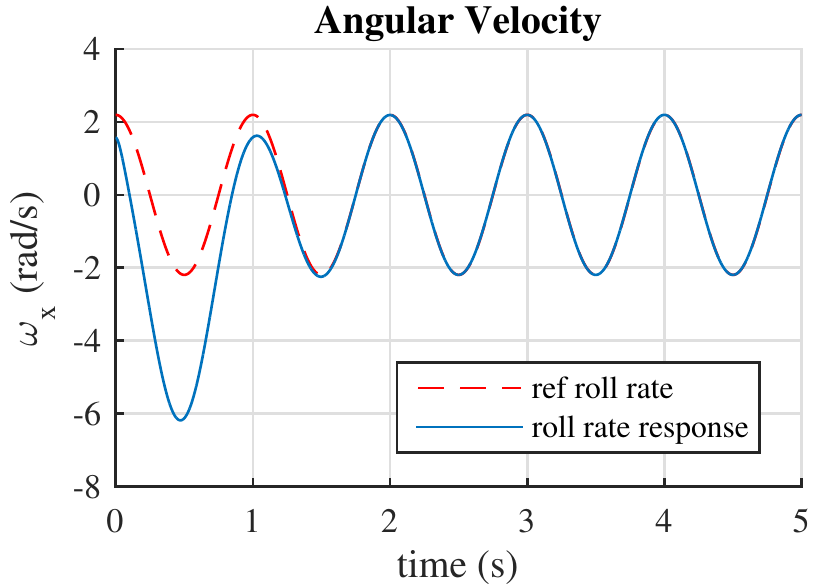}
  \caption{Roll rate tracking response.}
\label{fig:roll_response}
\end{figure}
\begin{figure} [!htbp]
\centering
  \includegraphics[width=0.9\linewidth]{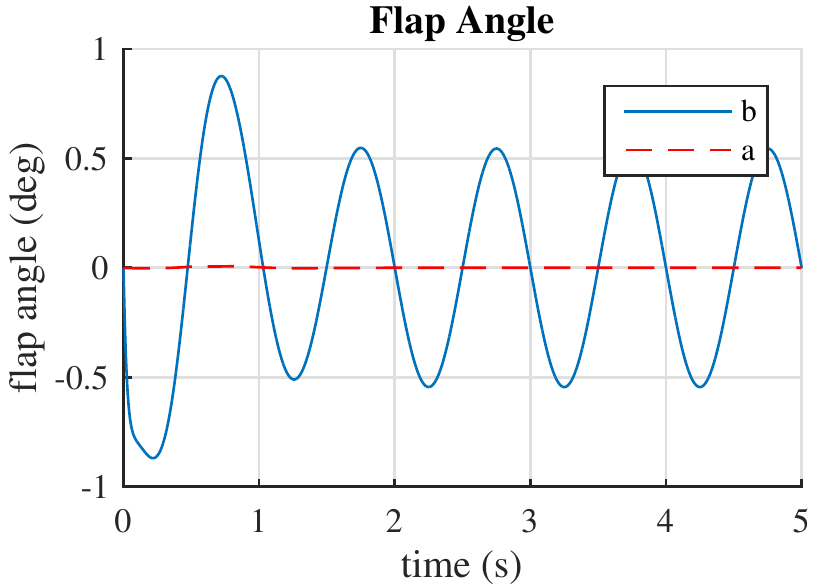}
  \caption{Required flap angle for tracking.}
  \label{fig:flap_angles}
\end{figure}

\begin{table} [h!]
\centering
\begin{tabular}{|c| c|} 
 \hline
 Parameter & Values  \\
 \hline
 $[J_{xx} J_{yy} J_{zz}]$ ($kg$-$m^2$) & [0.095 0.397 0.303]  \\ 
 \hline
 $\tau_m$ ($s$) & 0.06  \\
 \hline
 $k_\beta$ ($N$-$m$-$rad^{-1}$) & 129.09  \\
 \hline
 $h$ ($m$) & 0.174 \\
 \hline
 $K_\beta$ ($N$-$m$-$rad^{-1}$) & 137.7 \\ [1ex] 
 \hline
\end{tabular}
\caption{Helicopter parameters}
\label{tab:heli_params}
\end{table}


\section{Experimental Results}
The proposed controller was validated on an instrumented 10 kg class small scale conventional helicopter which consists of a single main rotor and a tail rotor. The instrumented helicopter is shown in Fig. \ref{fig:helicopter}. The main rotor of diameter 1.4 meter operates at 1500 rpm. The lateral and longitudinal control moment is produced by tilting the swashplate using three servos. Yawing moment is generated by changing the collective pitch of tail rotor. The helicopter has a stiff rotor hub (large $k_\beta$) which makes it extremely agile.

\begin{figure}[!htbp]
\centering
  \includegraphics[width=1.0\linewidth]{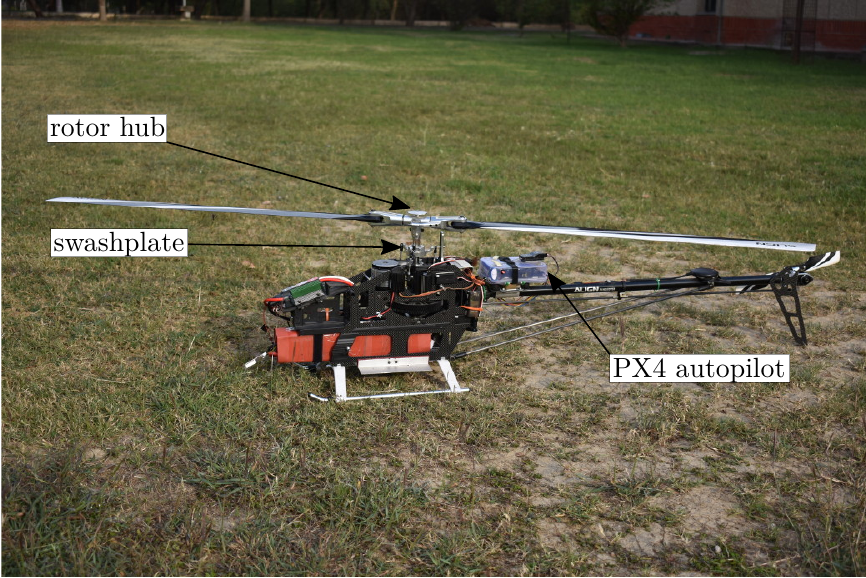}
  \caption{Experimental helicopter.}
  \label{fig:helicopter}
\end{figure}

The controller was implemented on PX4 autopilot hardware. It consists of a suite of sensors, namely 3-axis accelerometer, 3-axis gyroscope, 3-axis magnetometer, a GPS receiver and a barometer which together constitute the attitude heading reference system (AHRS). The autopilot software is based on PX4 flight stack which has a modular design and runs on top of a real-time operating system (NuttX). The autopilot comes with an EKF based attitude estimator. The proposed attitude controller was added as a module and runs at 250 Hz.

For validation purpose the helicopter was excited about roll/lateral axis. This axis of excitation was chosen as it eases the pilot to keep track of translation motion which is not the case with pitch/longitudinal movement. The input reference signal was a superposition of manual pilot input and autopilot generated sinusoidal roll reference input of $\pm$20 degree at 1 Hz. The manual input was superimposed as a correction so as to keep the translational motion of helicopter within a safe region.  The performance of the attitude controller was found to be satisfactory as seen in the linked video \cite{exp_youtube_link}. The error in tracking can be attributed to uncertainty in model structure and parameters.

\begin{figure}[!htbp]
\centering
  \includegraphics[width=0.9\linewidth]{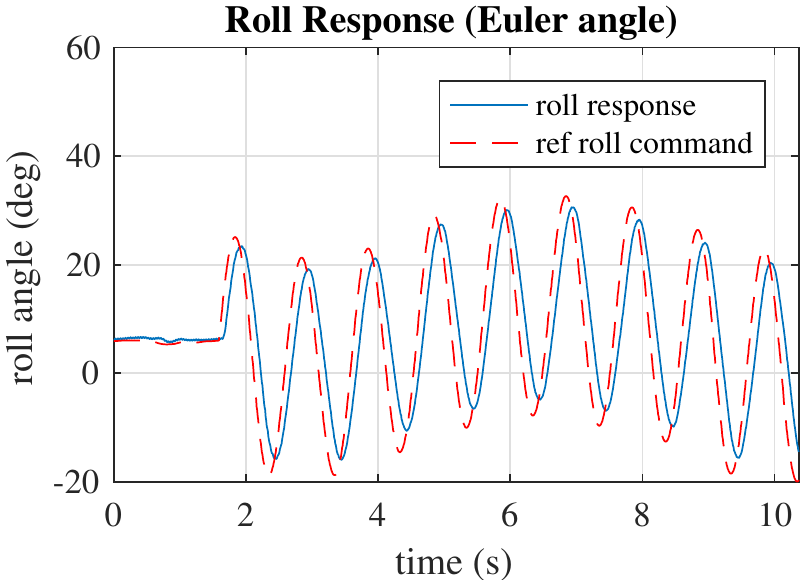}
  \caption{Experimental roll tracking response.}
  \label{fig:exp_roll_response}
\end{figure}

\section{Conclusion}
To the best of the knowledge of the authors, this work is the first attempt to integrate geometric control theory for the purpose of synthesizing an attitude tracking controller for a small-scale aerobatic helicopter, preserving the significant dynamics of the system while doing so.  The control law was validated in simulation and experiment on a 10 kg class small scale helicopter. The results as seen through the experimental validation are very encouraging.

\section{Acknowledgements}
Nidhish Raj and Ravi N Banavar acknowledge with pleasure, the support, and the conducive and serene
surroundings of IIT-Gandhinagar, where most of the theoretical work for this effort was 
done.

\nocite{*}
\bibliographystyle{unsrt}
\bibliography{ref}
\end{document}